\documentclass{bmcart}

\usepackage[utf8]{inputenc} 

\usepackage[english]{babel}
\usepackage{url}
\usepackage{amssymb}
\usepackage[cmex10]{amsmath}
\usepackage{bbm}
\usepackage{graphicx}
\usepackage{subcaption}
\usepackage{amsthm}

\startlocaldefs
\newtheorem{theorem}{Theorem}[section]
\newtheorem{lemma}[theorem]{Lemma}
\newtheorem{corollary}[theorem]{Corollary}
\endlocaldefs

\begin{document}

\begin{frontmatter}

\begin{fmbox}
\dochead{Research}


\title{Factorization threshold models for scale-free networks generation}


\author[
   addressref={aff1,aff3},
   email={artikov.akmalzhon@phystech.edu
 Authors are listed in alphabetical order.}
]{\inits{AA}\fnm{Akmal} \snm{Artikov}}
\author[
   addressref={aff1},
   email={ruskagerot@gmail.com}
]{\inits{AD}\fnm{Aleksandr} \snm{Dorodnykh}}
\author[
   addressref={aff1,aff2},
   email={kashin.yana@gmail.com}
]{\inits{YK}\fnm{Yana} \snm{Kashinskaya}}
\author[
   addressref={aff3},                   
   email={sameg@yandex-team.ru, samosvat.egor@gmail.com}   
]{\inits{SE}\fnm{Egor} \snm{Samosvat}}

\address[id=aff1]{
  \orgname{Moscow Institute of Physics and Technology (SU)}, 
  \city{Moscow},                              
  \cny{Russia}                                    
}
\address[id=aff2]{%
  \orgname{Skolkovo Institute of Science and Technology},
  \city{Moscow},                              
  \cny{Russia}    
}

\address[id=aff3]{
  \orgname{Yandex}, 
  \city{Moscow},                              
  \cny{Russia}  
}


\end{fmbox}


\begin{abstractbox}

\begin{abstract} 
In this article we suggest a new approach for scale-free networks generation with an alternative source of the power-law degree distribution. It comes from matrix factorization methods and geographical threshold models that were recently proven to show good results in scale-free networks generation.

We associate each node with a latent features vector distributed over a unit sphere and with a weight variable sampled from a Pareto distribution.  We join two nodes by an edge if they are spatially close and/or have large weights. The network produced by this approach is scale-free and has  a power-law degree distribution with an exponent of 2. In addition, we propose an extension of the model that allows us to generate directed networks with tunable power-law exponents.

\end{abstract}


\begin{keyword}
\kwd{scale-free networks}
\kwd{matrix factorization}
\kwd{threshold models}
\end{keyword}


\end{abstractbox}
%

\end{frontmatter}




\section{Introduction}
\indent \indent Most social, biological, topological and technological networks display distinct non-trivial topological features demonstrating that connections between the nodes are neither regular nor random at the same time \cite{albert2002statistical}. Such systems are called \textit{complex networks}. On of the well-known and well-studied classes of complex networks is \textit{scale-free networks} whose degree distribution $P(k)$ follows a power law $P(k) \sim k^{-\alpha}$, where $\alpha$ is a parameter whose value is typically in the range $2 < \alpha < 3$. Many real networks have been reported to be scale-free \cite{clauset2009power}.

Generating scale-free networks is an important problem because they usually have useful properties such as high clustering \cite{colomer2012clustering}, robustness to random attacks \cite{callaway2000network} and easy achievable synchronization \cite{moreno2004synchronization}. Several models for producing scale-free networks have been suggested; most of them are based on the preferential attachment approach \cite{albert2002statistical}. This approach forces existing nodes of higher degrees to gain edges added to the network more rapidly in a “rich-get-richer” manner. This paper offers a model with another explanation of scale-free property. 

Our approach is inspired by \textit{matrix factorization}, a machine learning method being successfully used for link prediction \cite{menon2011link}. The main idea is to approximate a network adjacency matrix by a product of matrices $V$ and $V^T$, where $V$ is the matrix of nodes' latent features vectors.  To create a generative model of scale-free networks we sample latent features $V$ from some probabilistic distribution and try to generate a network adjacency matrix. Two nodes are connected by an edge if the dot product of their latent features exceeds some threshold. This threshold condition is influenced by the \textit{geographical threshold models} that are applied to scale-free network generation \cite{hayashi2005review}. Because of the methods used (adjacency matrix factorization and threshold condition) we call our model the \textbf{factorization threshold model}.

A network produced in such a way  is scale-free and follows power-law degree distribution with an exponent of $2$, which differs from the results for basic preferential attachment models \cite{barabasi1999emergence, bollobas2001degree, holme2002growing} where the exponent equals $3$. We also suggest an extension of our model that allows us to generate directed networks with a tunable power-law exponent.

This paper is organized as follows. Section~2 provides information about related works that inspired us. The formal description of our model in the case of an undirected fixed size network is presented in Section~3, which is followed by a discussion of how to generate growing networks. In Section~4 the problem of making resulting networks sparse is considered. Section~5 shows that our model indeed produces scale-free networks. Extensions of our model, which allows to generate directed networks with a tunable power-law exponents and some other interesting properties, will be discussed in Section~6. Section~7 concludes the paper.

\section{Related work}
\indent \indent In this section we consider related works that encouraged us to create a new model for complex networks generation.

\subsection{Matrix Factorization}

\indent \indent \textit{Matrix factorization} is a group of algorithms where a given matrix $R$ is factorized into two smaller matrices $Q$ and $P$ such that: $R \approx Q^TP$ \cite{lee2001algorithms}. 

There is a popular approach in \textit{recommendation systems} which is based on matrix factorization \cite{koren2009matrix}. Assume that users express their preferences by rating some items, this can be viewed as an approximate representation of their interests.  Combining known ratings we get partially filled matrix $R$, the idea is to approximate unknown ratings using matrix factorization $R \approx Q^TP$. A geometrical interpretation is the following. The rows of matrices $Q$ and $P$ can be seen as latent features vectors $\vec{q}_i$ and $\vec{p}_u$ of items and users, respectively. The dot product $(\vec{q}_i, \vec{p}_u)$ captures an interaction between an user $u$ and an item $i$ and it should approximate the rating of the item $i$ by the user $u$: $R_{ui}  \approx (\vec{q}_i, \vec{p}_u)$. Mapping of each user and item to latent features is considered as an optimization problem of minimizing distance between $R$  and $Q^TP$  that is usually solved using SGD (stochastic gradient descent) or ALS (alternating least squares) methods.

Furthermore, matrix factorization was suggested to be used for link prediction in networks \cite{menon2011link}. Link prediction refers to the problem of finding missing or hidden links which probably exist in a network \cite{liben2007link}. In\cite{menon2011link} it is solved via matrix factorization: a network adjacency matrix $A$ is approximated by a product of the matrices $V$ and $V^T$, where $V$ is the matrix of nodes' latent features. 

\subsection{Geographical threshold models}
\indent \indent \textit{Geographical threshold models} were recently proven to have good results in scale-free networks generation \cite{hayashi2005review}. We are going to briefly summarize one variation of these models \cite{masuda2005geographical}.

Suppose the number of nodes to be fixed. Each node carries a randomly and independently distributed weight variable $w_i \in \mathbb{R}$. Also, the nodes are uniformly and independently distributed with specified density in a $\mathbb{R}^d$. A pair of nodes with weights $w, w'$ and Euclidean distance $r$ are connected if and only if:

\begin{equation}\label{eq:GmodelEq}
(w + w') \cdot h(r) \geq \theta,
\end{equation}
\\
where $\theta$ is the model threshold parameter and $h(r)$ is a distance function that is assumed to decrease in $r$ . For example, we can take $h(r) = r^{-\beta}$ , where $\beta > 0$ .

First, exponential distribution of weights with the inverse scale parameter $\lambda$ has been studied.  This distribution of weights leads to scale-free networks with a power-law exponent of 2: $P(k) \propto k^{-2}$. It is interesting that the exponent of a power-law does not depend on the $\lambda$, $d$ and $\beta$ in this case. Second, Pareto weight distribution with scale parameter $w_0$ and shape parameter $a$ has been considered. In this case a tunable power-law degree distribution has been achieved: $P(k) \propto k^{-1 - \frac{a \beta}{d}  }$. 

There are other variations of this approach: uniform distribution of coordinates in the $d-$dimensional unit cube \cite{morita2006crossovers}, lattice-based models \cite{rozenfeld2002scale}, \cite{warren2002geography} and even  networks embedded in fractal space \cite{yakubo2011scale}.

\section{Model description}
\indent \indent We studied theoretically matrix factorization by turning it from a trainable supervised model into a generative probabilistic model. When matrix factorization is used in machine learning the adjacency matrix $A$ is given and the goal is to train the model by tuning the matrix of latent features $V$ in such way that $A \approx V^T V$. In our model  we make the reverse: latent features $V$ are sampled from some probabilistic distribution and we generate a network adjacency matrix $A$ based on $V^T V$. 

Formally our model is described in the following way:
$$
\begin{cases}\label{eq:modelEq}
A_{ij} =  \mathrm{I}\left[( \vec{v_i},  \vec{v_j}) \geq \theta \right]\\
\vec{v_i} = w_i \vec{x_i}  \in \mathbb{R}^d \\
w_i  \sim \text{Pareto}(a, w_0) , ~ \vec{x_i} \sim \text{Uniform} ( S^{d-1}) \\
 i = 1\ldots n, ~ j = 1\ldots n 
\end{cases}
$$

\begin{itemize}
    \item Network has $n$ nodes and each node is associated with a $d$-dimensional latent features vector $\vec{v_i}$.
    \item Each latent features vector $\vec{v_i}$ is a product of weight $w_i$ and direction $\vec{x_i}$.
\item Directions $\vec{x_i}$ are i.i.d. random  vectors uniformly distributed over the surface of $(d-1)$-sphere.
    \item Weights are i.i.d. random variables distributed according to Pareto distribution with the following density function $f(w)$:
\begin{equation}\label{eq:pareto}
f(w) = \frac{a}{w_0} {(\frac{w_0}{w})}^{a + 1}\; (w \geq w_0).
\end{equation}
	\item Edges between nodes $i$ and $j$ appear if a dot product of their latent features vectors $(\vec{v_i}, \vec{v_j})$ exceeds a threshold parameter $\theta$.
\end{itemize}

Therefore, we take into consideration both node's importance $w_i$ and its location $x_i$ on the surface of a $(d-1)-$sphere (that can be interpreted as the Earth in the case of $\vec{x_i} \in S^{2} \subset \mathbb{R}^3$). Thus, inspired by the matrix factorization approach we achieved the following model behavior: the edges in our model are assumed to be formed when a pair of nodes is spatially close and/or has large weights. Actually, compared with the geographical threshold models we use dot product to measure proximity of nodes instead of Euclidean distance.

We have defined our model for fixed size networks but in principle our model can be generalized for the case of growing networks. The problem is that a fixed threshold $\theta$ when the size of a network tends to infinity with high probability leads to a complete graph. But real networks are usually sparse. 

Therefore, in order to introduce \textit{growing factorization threshold models} we use a threshold function $\theta := \theta(n)$ which depends on the number of nodes $n$ in the network. Then for every value of network size $n$ we have the same parameters except of threshold $\theta$. This means that at every step, when a new node will be added to the graph, some of the existing edges will be removed. In the next section we will try to find threshold functions which lead to sparse networks.

In order to preserve readability of the proofs we consider only the case $d = 3$ because proofs for higher dimensions can be derived in a similar way. However, we will give not only mean-field approximations but also strict probabilistic proofs, which to the best of our knowledge have not been done for geographical threshold models yet and can be likely applied in the other works too. 

\section{Generating sparse networks}
\indent \indent The aim of this section is to model sparse growing networks. To do this we need to find a proper threshold function.

First, we have studied the growth of the real networks.  For example, Figure~\ref{ris:Citation} shows the growth of a citation graph. The data was obtained from the SNAP\footnote{https://snap.stanford.edu/data/} database. It can be seen that the function $y(x) = 4.95 x \log x - 40 x$ is a good estimation of the growth rate of this network. That is why we decided to focus on the linearithmic or sub-linearithmic growth rate of the model (here and subsequently, by the growth of the model we mean the growth of the number of edges).

\begin{figure}[ht]
\center{\includegraphics[scale=0.5]{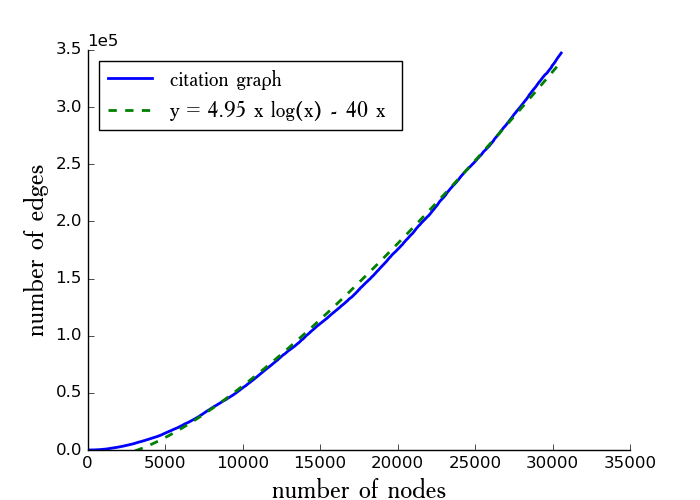}}
\caption{The growth of citation graph Arxiv HEP-PH}
\label{ris:Citation}
\end{figure}

\subsection{Analysis of the expected number of edges}

Let $M(n)$ denote the number of edges in the network of size $n$. To find its expectation we need the two following lemmas. 

\begin{lemma}\label{lm:probWithWeight}
The probability for a node with weight $w$ to be connected to a random node is
\begin{equation}\label{eq:p-w-e}
P_e(w)  =
 \begin{cases}
   \frac{1}{2}\left(1 - \frac{a\theta}{w(a+1)w_{0}}\right), &w > \frac{\theta}{w_0},\\
   \frac{1}{2}\frac{w_{0}^a}{\theta^a (a + 1)} w^a, &w \leq \frac{\theta}{w_0}.
 \end{cases}
\end{equation}
\end{lemma}

\begin{lemma}\label{lm:expectedNumberEdges}
The edge probability in the network is
  \begin{equation}\label{eq:expectedNumberEdges}
P_e  =
 \begin{cases}
   \frac{1}{2} -  \frac{1}{2} \frac{a^2}{(a+1)^2}\frac{\theta}{w_0^{2}}, &\theta < w_{0}^2,\\
   \frac{w_0^{2a}}{2 \theta^a} \Big(\frac{a (\ln \theta - 2 \ln w_0)}{a+1} - \frac{a^2}{(a+1)^2} + 1\Big), &\theta \geq w_{0}^2. 
 \end{cases}
\end{equation}
\end{lemma}

To improve readability, we moved the proofs of Lemma~\ref{lm:probWithWeight} and Lemma~\ref{lm:expectedNumberEdges} to Appendix.

The next theorem shows, that our model can have any growth which is less than quadratic.

\begin{theorem}\label{th:anygrowth}
Denote as $R(n)$ such function that $R(n) = o(n^2)$ and $R(n) > 0$. Then there exists such threshold function $\theta(n)$ that the growth of the model is $R(n)$: 
$$ \exists N \,\,\, \mathrm{E}M(n) = R(n) \,\,\,\,\, (n \geq N).$$
\end{theorem}
\begin{proof}

It easy to check that $P_e$ is a continuous function of $\theta$. The intermediate value theorem states that $P_e(\theta)$ takes any value between $P_e(\theta = 0) = 1/2$ and $P_e(\theta = \infty) = 0$ at some point within the interval.\\

Since  $R(n) = o(n^2)$ and positive, there exists $N$ such that for all $n \ge N$, $0 < R(n) < \frac{1}{2} \times \frac{n(n-1)}{2}$.\\

It means that the equation $\mathrm{E}M(n) = R(n)$ is feasible for all $n \ge N$.
\end{proof}

Taking into account Theorem~\ref{th:anygrowth}, we obtain parameters for the linearithmic and linear growths of the expected number of edges. 

\begin{theorem} \label{th:loglinear}
Suppose the following threshold function: ${\theta(n) = D n^{\frac{1}{a}}}$, where D is a constant. Then the growth of the model is linearithmic:
$$ \mathrm{E}M (n) = A n \ln n (1 + o(1)) \,\,\,\text{ } (n \geq \frac{w_0^{2a}}{D^a}) ,$$
 where $A$ is a constant depending on the Pareto distribution parameters.
\end{theorem}
\begin{proof}
We can rewrite inequality $n \geq \frac{w_0^{2a}}{D^a}$ as ${Dn^{\frac{1}{a}} \geq w_0^2}$ 
and apply Lemma~\ref{lm:expectedNumberEdges} in the case $\theta(n) = Dn^{\frac{1}{a}} \geq w_0^2$ 
\begin{equation}\label{eq:generalEM}
\mathrm{E}M = \frac{n(n-1)}{2}\frac{w_0^{2a}}{2 \theta^a} \Big(\frac{a (\ln \theta - 2 \ln w_0 )}{a+1} - \frac{a^2}{(a+1)^2} + 1\Big).
\end{equation}

If we replace $\theta$ by $Dn^{\frac{1}{a}}$, we obtain
\begin{gather*}
\mathrm{E}M(n) = \frac{n(n-1)w_0^{2a}}{4 (Dn^{\frac{1}{a}})^a} \Big(\frac{a (\ln (Dn^{\frac{1}{a}}) - 2 \ln w_0)}{a+1} - \frac{a^2}{(a+1)^2} + 1\Big) = \\
= \frac{(n-1)w_0^{2a}}{4D^a} \Big(\frac{\ln n}{a+1} - \frac{a^2}{(a+1)^2} + 1 + \frac{a(\ln D - 2\ln w_0)}{a+1}\Big) = \\
=  A n \ln n (1 + o(1)).
\end{gather*}
\end{proof}

\begin{theorem} \label{th:linear}
Suppose that the growth of the model is sub-linearithmic: ${\frac{\mathrm{E}M(n)}{n\ln n} = o(1)}$, then ${\frac{n^{\frac{1}{a}}}{\theta(n)} = o(1)}$.
\end{theorem}
\begin{proof}
Let us consider another model with a threshold function $\theta'(n) = Dn^{\frac{1}{a}}$ and the expected number of edges $\mathrm{E}M'(n)$. According to Theorem~\ref{th:loglinear} and the condition ${\frac{\mathrm{E}M(n)}{n\ln n} = o(1)}$ there exists a natural number $N_D$ such that 
$$\forall n \geq  N_D  \,\,\,\,\, \mathrm{E}M'(n) = A n \ln n (1 + o(1)) \geq  \mathrm{E}M(n).$$
  
This also means that for all $n \geq  N_D$ we have $\theta(n) \geq \theta'(n)$. Therefore
$$\forall n \geq  N_D \,\,\,\,\,  \frac{n^{\frac{1}{a}}}{\theta(n)} \leq \frac{n^{\frac{1}{a}}}{\theta'(n)} = \frac{1}{D}. $$

By the arbitrariness of the choice of $D$, we have $ \frac{n^{\frac{1}{a}}}{\theta(n)} = o(1)$.
\end{proof}

\subsection{Concentration theorem}
\indent \indent In this section we will find the variance of the number of the edges and prove the concentration theorem 

Proofs of the following lemmas can be found in the appendix. 

\begin{lemma} \label{th:variation}
Suppose that $x$, $y$ and $z$ are random nodes. Let $P_{<}$ be the probability for the node $x$ to be connected to both nodes $y$ and $z$. 
Then the variance of the number of edges $M$ is 
\begin{equation*}
\mathrm{Var}(M) = \frac{n(n-1)}{2} P_e(1 - P_e) + n \frac{(n-1)(n-2)}{2}(P_{<} - P_e^2),
\end{equation*}
\end{lemma}

\begin{lemma}\label{lm:P_triangle}
Suppose that $x$, $y$ and $z$ are random nodes. Let $P_{<}$ be the probability for the node $x$ to be connected to both nodes $y$ and $z$. \\
Then
\begin{equation*}
P_{<} = 
\begin{cases}
\frac{1}{4}\frac{w_0^{2a}}{\theta^{2a}(a+1)^2} [\theta^{a} - w_0^{2a}] + \frac{1}{4}\frac{w_0^{2a}}{\theta^{a}}\Big[ 1 - 
2 \frac{a^2 }{(a+1)^2} + \\ +
\frac{a^3   }{(a+1)^2(a+2)} \Big], &\theta \geq w_{0}^2,\\
\frac{1}{4}  -
\frac{1}{2} \frac{a^2 \theta }{(a+1)^2} \frac{1}{w_0^2} +
\frac{1}{4}\frac{a^3 \theta^2 }{(a+1)^2(a+2)} \frac{1}{w_0^4} , &\theta < w_{0}^2.
\end{cases}
\end{equation*}
\end{lemma}

Combining these results, we get the following theorem, that will be needed to prove the concentration theorem
\begin{theorem}\label{lm:concetrationHelp}
If $\theta \geq w_0^2$, the variance is
\begin{equation*}
\mathrm{Var}(M) = \mathrm{E}M + n \frac{(n-1)(n-2)}{2}\Big[A\frac{1}{\theta^{a}} + B \frac{1}{\theta^{2a}} \Big] -  \frac{2(n-2)}{n(n-1)} (\mathrm{E}M)^2 ,
\end{equation*}
 where $A$ and $B$ are constants which depend on the Pareto distribution parameters. 
\end{theorem}
\begin{proof}
According to Lemma~\ref{th:variation} and Lemma~\ref{lm:P_triangle} in case of $\theta \geq w_0^2$, the variance is 
\begin{gather} \label{eq:concentraion-variation-2}
\mathrm{Var}(M) = \frac{n(n-1)}{2} P_e(1 - P_e) + n \frac{(n-1)(n-2)}{2}(P_{<} - P_e^2). \\
P_{<} = \frac{1}{4}\frac{w_0^{2a}}{\theta^{2a}(a+1)^2} [\theta^{a} - w_0^{2a}] + \frac{1}{4}\frac{w_0^{2a}}{\theta^{a}}\Big[ 1 - 
2 \frac{a^2 }{(a+1)^2}  + 
\frac{a^3   }{(a+1)^2(a+2)} \Big]
\end{gather}

According to Lemma~\ref{lm:expectedNumberEdges}, the expected number of edges is 
\begin{equation} \label{eq:concentraion-expectedEdges}
\mathrm{E}M = \frac{n(n-1)}{2}P_e.
\end{equation}

Combining (\ref{eq:concentraion-expectedEdges}) and (\ref{eq:concentraion-variation-2}), we obtain 
\begin{gather*}
\mathrm{Var}(M) = \mathrm{E}M(1 - P_e) + n \frac{(n-1)(n-2)}{2}P_{<} - \mathrm{E}M (n-2) P_e = \mathrm{E}M + \\+ n \frac{(n-1)(n-2)}{2}P_{<} -  \frac{2(n-2)}{n(n-1)} (\mathrm{E}M)^2.
 \end{gather*}

Therefore,
\begin{gather*}
P_{<} = \frac{1}{4}\frac{w_0^{2a}}{\theta^{2a}(a+1)^2} [\theta^{a} - w_0^{2a}] + \frac{1}{4}\frac{w_0^{2a}}{\theta^{a}}\Big[ 1 - 
2 \frac{a^2 }{(a+1)^2} + \\ + 
\frac{a^3   }{(a+1)^2(a+2)} \Big] = \frac{1}{\theta^a}C_{1} - \frac{1}{\theta^{2a}}C_{2} + \frac{1}{\theta^a} C_{3} = A\frac{1}{\theta^{a}} + B \frac{1}{\theta^{2a}},
\end{gather*}

where $C_1$, $C_2$, $C_3$, $A$ and $B$ are constants depending on the Pareto distribution parameters. 

Finally, we obtain
\begin{gather*}
\mathrm{Var}(M) = \mathrm{E}M + n \frac{(n-1)(n-2)}{2}\Big[A\frac{1}{\theta^{a}} + B \frac{1}{\theta^{2a}} \Big] -  \frac{2(n-2)}{n(n-1)} (\mathrm{E}M)^2.
\end{gather*}
\end{proof}

\begin{theorem}{Concentration theorem}\label{th:Con}\\
If $\theta(n)$ and $\mathrm{E}M(n)$ tends to infinity as $n \to \infty$ and $\frac{n^3}{(\mathrm{E}M(n))^2\theta(n)^a} = o(1)$, then      
$$ \forall \varepsilon > 0 \,\,\, P(|M - \mathrm{E}M| \geq \varepsilon \cdot \mathrm{E}M) \xrightarrow[]{n\to \infty} 0 ,$$
where $M$ is the number of edges  in the graph.
\end{theorem}
\begin{proof}
According to Chebyshev's inequality, we have
\begin{equation}\label{eq:ChebLinear}
P(|M - \mathrm{E}M | \geq  \varepsilon \cdot \mathrm{E}M )\leq \frac{\mathrm{Var}(M)M}{\varepsilon^2 \cdot (\mathrm{E}M)^2}  .
\end{equation}
Let us estimate the right part of the inequality. 
Using Theorem~\ref{lm:concetrationHelp}, we get
$$ \frac{\mathrm{Var}(M)}{\varepsilon^2 \cdot (\mathrm{E}M)^2} = \frac{1}{\varepsilon^2 \mathrm{E}M} + \frac{O(n^3)}{(\mathrm{E}M)^2}\Big[ A\frac{1}{\theta^{a}} + B \frac{1}{\theta^{2a}} \Big] + O(\frac{1}{n}) = $$
$$ = \frac{1}{\varepsilon^2 \mathrm{E}M} + \frac{O(n^3)}{(\mathrm{E}M)^2}\frac{1}{\theta^{a}}\Big[ 1 + \frac{B}{A\theta^{2a}} \Big] + O(\frac{1}{n}).$$
Using the conditions of the theorem, we obtain
$$ \frac{\mathrm{Var}(M)}{\varepsilon^2 \cdot (\mathrm{E}M)^2} \to 0 \text{ as } n \to \infty.$$
\end{proof}

Combining Theorems \ref{th:loglinear}, \ref{th:linear} and \ref{th:Con} we obtain the following corollary.

\begin{corollary}\label{th:ConCorollary}
Suppose that one of the following conditions holds: 
\begin{itemize}
\item  the threshold function $\theta(n)$ equals $D n^{\frac{1}{a}}$
\item $\frac{n}{\mathrm{E}M(n)} = O(1)$ and $\frac{\mathrm{E}M(n)}{n\ln n} = o(1)$
\end{itemize}
Then
$$ \forall \varepsilon > 0 \,\,\, P(|M - \mathrm{E}M| \geq \varepsilon \cdot \mathrm{E}M) \xrightarrow[n\to \infty]{} 0 ,$$
where $M$ is the number of edges  in the graph.
\end{corollary}

In this way we have proved that the number of edges in the graph does not deviate much from its expected value. It means that having the linearithmic or the sub-linearithmic growth of the expected number of edges we also have the same growth for the actual number of edges.

\section{Degree distribution}

\indent \indent In this section we show that our model follows power-law degree distribution with an exponent of $2$ and give two proofs. The first is a mean-field approximation. It is usually applied for a fast checking of hypotheses. The second one is a strict probabilistic proof. To the best of our knowledge it has not been considered in the context of the geographic threshold models yet.

To confirm our proofs we carried out a computer simulation and plotted complementary cumulative distribution  of node degree which is shown on Figure~\ref{fig:powerlaw}.  
We also used a discrete power-law fitting method, which is described in \cite{clauset2009power} and implemented in the network analysis package igraph \footnote{http://igraph.org/}.
We obtained $\alpha = 2.16$, $x_{\min} = 4$  and a quite large $p$-value of $0.9984$ for the Kolmogorov-Smirnov  goodness-of-fit test.

\begin{figure}%
   \centering
   \begin{subfigure}[b]{0.7\textwidth}
       \includegraphics[width=\textwidth]{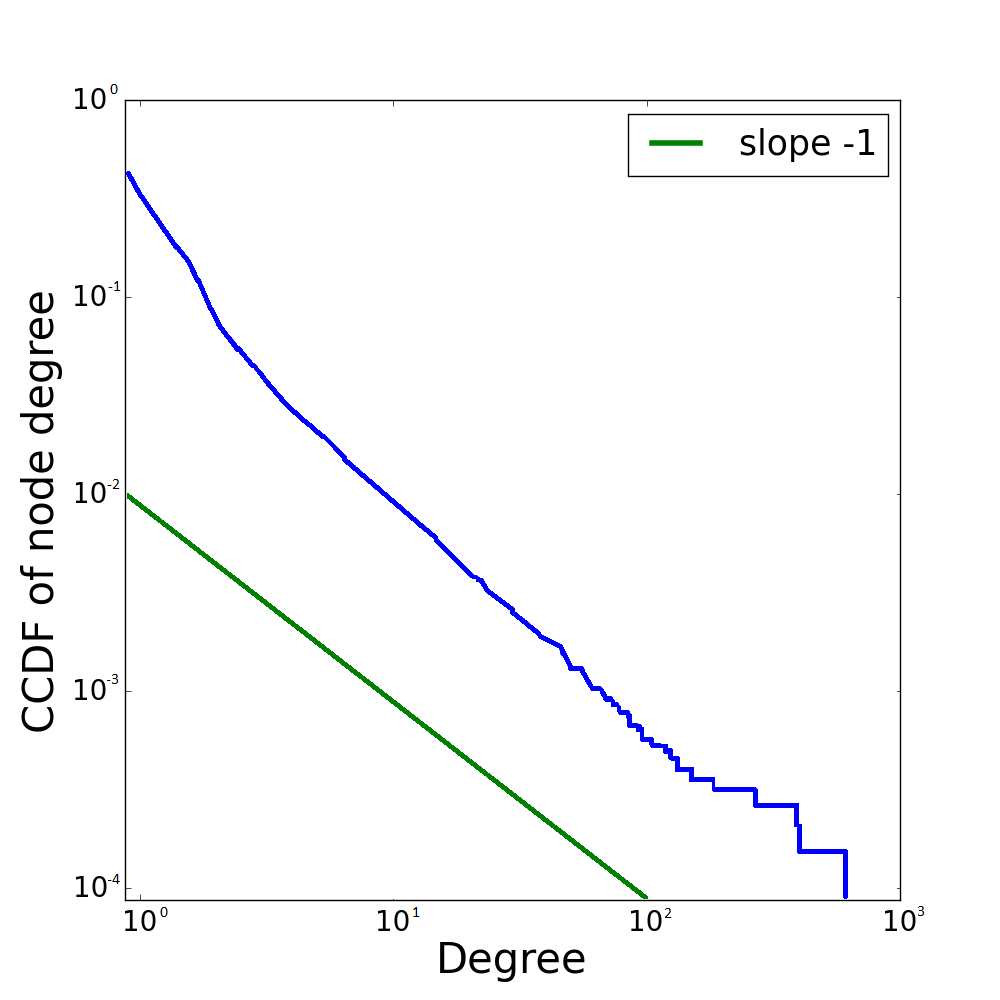}
   \end{subfigure}
   \caption{Complementary cumulative distribution of node degree \\ $n = 3 \cdot 10^5$, $\vec{x_i} \in \mathbb{R}^{3}$, $w_i \sim Pareto(3, 1)$, $\theta = 66.9$}
   \label{fig:powerlaw}
\end{figure}

\begin{theorem}\label{dd:original}
Let $P(k)$ be the probability of a random node to have a degree $k$. If $\frac{n^{\frac{1}{a}}}{\theta(n)} = o(1)$, then there exist such constants $C_0$ and $N_0$ such that $\forall~k(n):\forall~n~>~N_0 \,\, k(n)~<~C_0n$ we have
$$P(k) = (1+o(1)) k^{-2}.$$
\end{theorem}

\begin{proof}[Mean-field approximation] This approximation gives power-law only for nodes with weights $w \leq \frac{\theta}{w_0}$. But  the expected number of nodes with weights not satisfying this inequality  $\mathrm{E}m$  is extremely small

\begin{equation}
	\mathrm{E} m = n P(w > \frac{\theta}{w_0}) = n\left(\frac{w_0^2}{\theta}\right)^a = o(1). 
\end{equation}

As it was shown in Lemma \ref{lm:probWithWeight}, the probability of the node $\vec{v_i} = w_i \vec{x_i}$ with weight $w_i = w \leq \frac{\theta}{w_0}$ to have an edge to another random node is

$$
	P_e(w) = \frac{w_{0}^a}{2 \theta^a (a + 1)} w^a.
$$

Let $k_i(w)$ be the degree of the node $v_i$. Then 

$$
k_i(w) = \sum_{i \neq j} I[v_i \text{ is connected to } v_j],
$$
where $I$ stands for the indicator function.

As all nodes are independent, we get 
$$E k_i(w) = (n-1) P_e(w).$$

In the mean-field approximation we assume that $k_i(w)$ is really close to its expectation and we can substitute it by ${(n-1) P_e(w)}$ 
in the following expression for the degree distribution $P(k) = f(w) \frac{\mathrm{d}w}{\mathrm{d}k},$ where $f(w)$ is a density of weights. Thus, 

$$
	P(k) = \frac{2 a w_0^a\theta^a(a+1)}{(n-1) w^{2a}} \propto k^{-2}
$$
\end{proof}

Note that we have not used conditions on $k(n)$ and $\theta(n)$ yet, they are needed to estimate residual
terms in the following rigorous proof.

\begin{proof}[Proof]

Degree $k_i$ of the node $v_i$ is a binomial random variable. Using the probability $P_e(w)$ of the node $v_i$ with weight ${w_i = w}$ to have an edge to another random node, we can get the probability that $k_i$ equals $k$:
\begin{gather*}
P(k_i = k | w_i = w) = {n -1 \choose k} \left(P_e(w)\right)^k (1-P_e(w))^{n-k-1}.
\end{gather*}

To get the total probability we need to integrate this expression with respect to~$w$
\begin{gather*}
P(k_i = k) = {n - 1 \choose k} \int_{w_0}^{\infty} \left(P_e(w)\right)^k (1-P_e(w))^{n-k-1} \frac{aw_0^a}{w^{a+1}} \mathrm{d}w.
\end{gather*}

Because of $P_e(w)$ is a composite function, the integral breaks up into two parts.
\begin{gather*}
	I_1 = \int_{w_0}^{\theta/w_0} \left(P_e(w)\right)^k (1-P_e(w))^{n-k-1} \frac{aw_0^a}{w^{a+1}} \mathrm{d}w,\\
	I_2 = \int_{\theta/w_0}^{\infty} \left(P_e(w)\right)^k (1 - P_e(w))^{n-k-1} \frac{aw_0^a}{w^{a+1}}\mathrm{d}w.
\end{gather*}


Thus,
\begin{gather*}
P(k_i = k) =  {n - 1 \choose k} (I_1 + I_2).
\end{gather*}
For estimating $I_1$ we can use the formula $P_e(w) = \frac{1}{2}\frac{w_{0}^a}{\theta^a (a + 1)} w^a$ from Lemma~\ref{lm:probWithWeight}.
After making the substitution to integrate with respect to $P_e(w)$ and using the incomplete beta-function, we get
\begin{gather*}
I_1 = \frac{w_0^{2a}}{2\theta^a a(a+1)} \cdot 
\left(B \left(\frac{1}{2(a+1)}; k-1, n-k \right) - \right. \\ \left. - B \left(\frac{w_0^{2a}}{2\theta^a(a+1)}; k-1, n-k \right)\right).
\end{gather*}

For $I_2$ we can derive an upper bound. Note that for $w \geq \theta/w_0$  we have 
$$P_e(w) = \frac{1}{2}\left(1 - \frac{a\theta}{w(a+1)w_0}\right) < \frac{1}{2}$$ $$1 - P_e(w) \leq 1 - P_e(\theta / w_0) = \frac{1}{2}\left( 1 + \frac{a}{a+1} \right) = \varepsilon_0 < 1. $$
Therefore we obtain the following upper estimate

\begin{gather*}
	I_2 = O\left( \frac{ (\varepsilon_0)^{n-k-1}  }{2^k}   \int_{\theta/w_0}^{\infty}  \frac{aw_0^a}{w^{a+1}}\mathrm{d}w  \right) 
    = O\left( \frac{ (\varepsilon_0)^{n-k-1}}{\theta^a 2^k}   \right) 
\end{gather*}

We now combine estimates for $I_1$, $I_2$ and the following estimates for the incomplete beta-function:
\begin{gather*}
B(x; a, b) = O\Big(\frac{x^a}{a}\Big),\\
B(x; a, b) = B(a, b) + O\Big(\frac{(1-x)^b}{b}\Big),\\
\frac{1}{B(d-1, n-d)} = \frac{\Gamma(n-1)}{\Gamma(d-1)\Gamma(n-d)} = O\Big(\frac{n^{d-1}}{\Gamma(d-1)}\Big).
\end{gather*}
This gives us

\begin{gather*}
P(k_i = k) = \binom{n-1}{k} \frac{w_0^{2a}}{2 \theta^a a(a+1)}
\Big[B(k-1, n-k) +  O\Bigg( \frac{\left(1-\frac{1}{2(a+1)}\right)^{n-k}}{n - k}\Bigg) -\\- O\Bigg( \frac{\left(\frac{w_0^{2a}}{2\theta^a(a+1)}\right)^{k-1}}{k-1} \Bigg) + O\Bigg(  \frac{ (\varepsilon_0)^{n-k-1}}{\theta^a 2^k}   \Bigg) \Big] =\\ = 
\binom{n-1}{k} \frac{w_0^{2a}}{2 \theta^a a(a+1)} B(k-1, n-k) \Big[1+ O\Bigg( \frac{\left(\varepsilon_1\right)^{n-k} n^{k-1}}{(n - k)\Gamma(k-1)}\Bigg) \ + \\ \ + \ O\Bigg( \frac{\left(\frac{w_0^{2a}}{2\theta^a(a+1)}\right)^{k-1}n^{k-1}}{(k-1)\Gamma(k-1)} \Bigg) \ + \ O\Bigg(  \frac{ (\varepsilon_0)^{n-k-1}}{\theta^a 2^k}  \frac{n^{k-1}}{\Gamma(k-1)} \Bigg)\Big].
\end{gather*}

Let us introduce the following notations:

\begin{gather*}
 A = O\Bigg( \frac{\left(\varepsilon_1\right)^{n-k} n^{k-1}}{(n - k)\Gamma(k-1)}\Bigg), \text{where } \varepsilon_1 = 1 - \frac{1}{2(a+1)}, \\
 B = O\Bigg( \frac{\left(\frac{w_0^{2a}}{2\theta^a(a+1)}\right)^{k-1}n^{k-1}}{(k-1)\Gamma(k-1)} \Bigg),\\
 C = O\Bigg(  \frac{ (\varepsilon_0)^{n-k-1}}{\theta^a 2^k}  \frac{n^{k-1}}{\Gamma(k-1)} \Bigg), \text{where } \varepsilon_0 =\frac{1}{2}\left( 1 + \frac{a}{a+1} \right).
\end{gather*}

Using $\frac{n}{\theta^a(n)} = o(1)$, for $k(n) < C_0n$ we get 

\begin{gather*}
	B = O\Bigg( \frac{\left(\frac{w_0^{2a}}{2(a+1)}\right)^{k-1}(\frac{n}{\theta^a})^{k-1}}{\Gamma(k)} \Bigg) = o(1).
\end{gather*}

If $k(n)$ is a bounded function, then since $\varepsilon_0 < 1$ and $\varepsilon_1  < 1$ we have

\begin{gather*}
	A = O\left( \left(\varepsilon_1\right)^\frac{n-k}{k-1} n^{k-1}\right) = o(1),\\
    C = O\left( \left(\varepsilon_0\right)^{n-k} n^{k-1} \right) = o(1).
\end{gather*}

If $k(n) \rightarrow \infty$ as $n \rightarrow \infty$, using Stirling's approximation $\Gamma(k-1) \sim \sqrt{2 \pi (k-2)} \left( \frac{e}{k-2}\right)^{k-2}$ we get

\begin{gather*}
	A = O \left( \frac{k - 2}{(n-k) \sqrt{k-2}} \left((\varepsilon_1)^\frac{n-k}{k-1}  \frac{n}{k - 2}\right)^{k-1}\right),\\
    C = O \left( \frac{\sqrt{k-2}}{\theta^a} \left(  (\varepsilon_0)^{\frac{n-k-1}{k-1}}  \frac{n}{k-2} \right)^{k-1} \right).
\end{gather*}

Since  $ \varepsilon^x  x \rightarrow 0$  for $\varepsilon < 1$ as $x \rightarrow \infty$ there exist constants $C_0$ and $N_0$ such that for $n > N_0$  and $k(n) < C_0n$ we have $(\varepsilon_1)^\frac{n-k}{k-1}  \frac{n}{k - 2} < 1$ and $(\varepsilon_0)^{\frac{n-k-1}{k-1}}  \frac{n}{k-2} < 1$. This implies that $A=o(1)$ and $C=o(1)$.

Thus, we obtain

%
%
%
%
%
\begin{equation}
	P(k_i = k) =(1 + o(1))  {n - 1 \choose k} B(k - 1, n - k) = (1 + o(1))  k^{-2}.
\end{equation}
\end{proof}

%
%



Note that  regardless of the shape parameter of the Pareto distribution of weights we always generate networks with a degree distribution following a power law with an exponent equals $2$. In the next section we modify our model in order to change the exponent of the degree destribution and some other properties of the resulting networks.

\section{Model modifications}

\indent \indent In this section we will show how to modify our model to get new properties and how these modifications will affect the degree distribution.
\subsection{Directed network}

\indent \indent Many real networks are directed. In order to model them and obtain an exponent of the power law that differs from $2$, we changed the condition for the existence of an edge. There will be a directed edge $(v_i, v_j)$, if and only if

$$
	(w_i^{\alpha} \vec{x_i}, w_j^{\beta} \vec{x_j}) \geq \theta, \alpha,\beta > 0.
$$

As it follows from the next theorem this modification allows us to tune an exponent of the power law.

\begin{theorem}\label{th:ordd}
	Let $P_{out}(k)$ be the probability of an random node to have out-degree $k$, $P_{in}(k)$ – in-degree $k$. If ${n^{\max\{\alpha, \beta\}/a}/\theta(n) = o(1)}$, then there exist constants $C_0$ and $N_0$ such that $\forall k(n):\forall n > N_0 \,\, k(n) < C_0n$ we have 
$$
	P_{out}(k) = (1+o(1)) k^{-1 - \alpha/\beta},
	P_{in}(k) = (1+o(1))k^{-1 - \beta/\alpha}.
$$	
\end{theorem}

\begin{proof} 
	Here is a proof for the out-degree distribution. The case of the in-degree distribution is similar.
	
    Firstly, let us compute $P_e(w)$ – the probability of the node $\vec{v_i} = w_i \vec{x_i}$ with weight $w_i = w$ to have an edge to another random node.
    
    \begin{equation}
		P_e(w) = \int_{w_0}^{\infty}f(w')\int_{\substack{x' \in S(0,1) \\ (w^{\alpha} x, (w')^{\beta} x')\geq\theta}} \frac{1}{4\pi} \mathrm{d}x' \mathrm{d}w'.
	\end{equation}
    
    Similarly to Lemma ~\ref{lm:probWithWeight} we get
    \begin{equation}
		P_e(w) = \int_{\max\{w_0, \theta^{1/\beta} / w^{\alpha/\beta}\}}^{\infty}\frac{a w^a_0}{(w')^{a+1}} \frac{1}{2} \left(1 - \frac{\theta}{w^\alpha (w')^\beta}\right) \mathrm{d}w'.
	\end{equation}
    
Thus, we obtain
    \begin{equation}
		P_e(w)  =
         \begin{cases}
           \frac{1}{2}\left(1 - \frac{a\theta}{w^\alpha(a+\beta)w_{0}^\beta}\right), &w > \left(\frac{\theta}{w_0^\alpha}\right)^{1/\beta},\\
           \frac{w^{a\alpha/\beta}w_0^a}{2\theta^{a/\beta}}\left(\frac{1}{a} - \frac{1}{\beta + a} \right), &w \leq \left(\frac{\theta}{w_0^\alpha}\right)^{1/\beta}.
         \end{cases}
	\end{equation}
    
%
%

Like in Theorem ~\ref{dd:original} we have

\begin{gather*}
	P(k_i = k) = {n - 1 \choose k} \int_{w_0}^{\infty} \left(P_e(w)\right)^k (1-P_e(w))^{n-k-1} \frac{aw_0^a}{w^{a+1}} \mathrm{d}w.
\end{gather*}



The rest of the proof is similar to the corresponding steps of Theorem~\ref{dd:original}, so we omit details here.

\end{proof}

With $\alpha = \beta$ this model turns into the undirected case with the power law exponent equals $2$ that agrees with Theorem ~\ref{dd:original}. 

 

\subsection{Functions of dot product}

\indent \indent In our model because of the condition ${w_i w_j (\vec{x_i}, \vec{x_j}) \geq \theta \geq 0}$ node $\vec{v_i}$ can only be connected to the node $\vec{v_j}$ if an angle between $\vec{x_i}$ and $\vec{x_j}$ is less than $\pi/2$. This is a constraint on the possible neighbours of a node that restricts the scope of our model. 

We can solve this issue by changing the condition for the existence of an edge: 
\begin{equation}\label{ineq:h_condition}
 w_i^{\alpha} w_j^{\beta} h((\vec{x_i},\vec{x_j})) \geq \theta,
\end{equation}
where $h:[-1, 1] \to \mathbb{R}$.
On Figure~\ref{fig:r2dif} is an example of how it works in $\mathbb{R}^2$.

\begin{figure}%
   \centering
   \begin{subfigure}[b]{0.3\textwidth}
       \includegraphics[width=\textwidth]{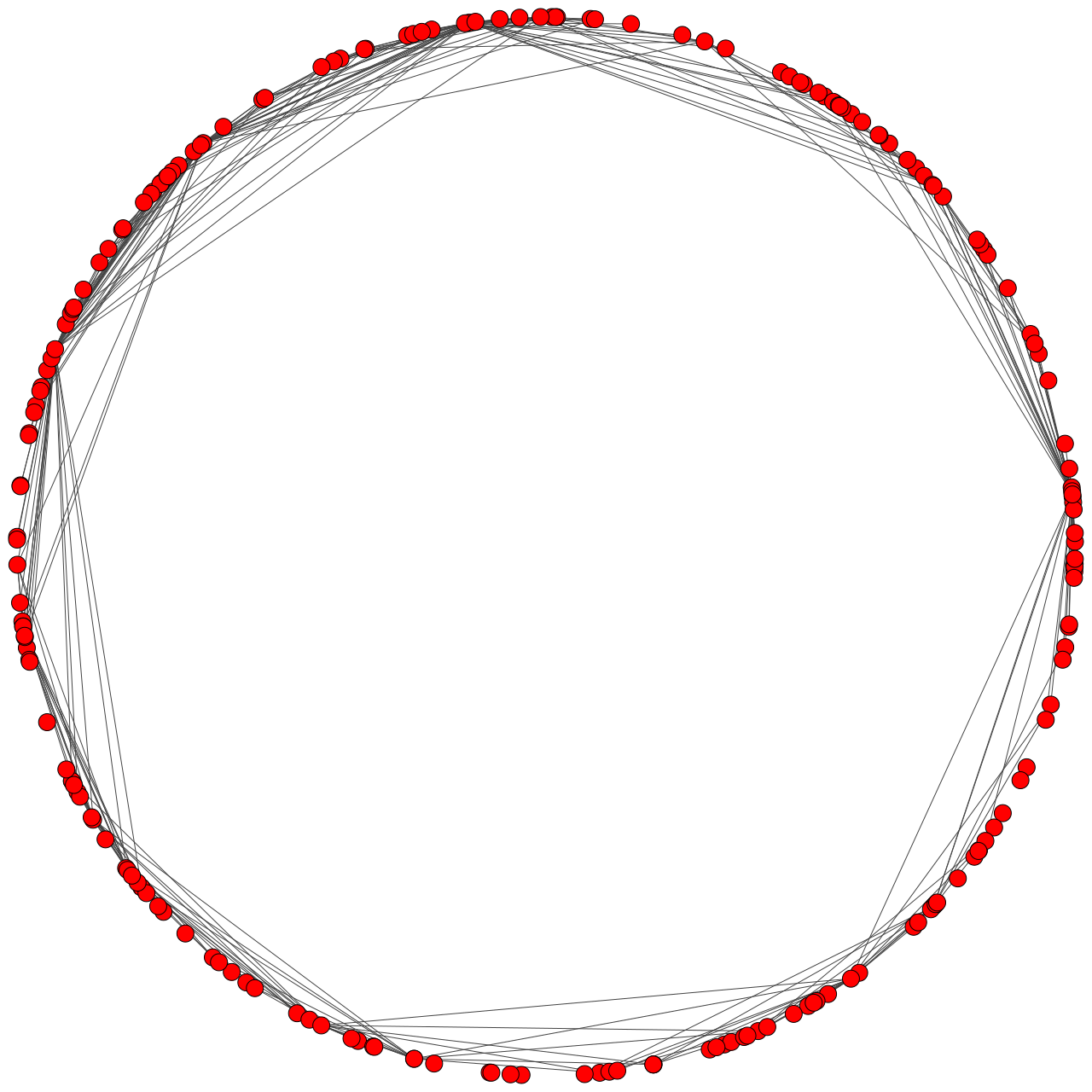}
   \end{subfigure}
   \begin{subfigure}[b]{0.3\textwidth}
       \includegraphics[width=\textwidth]{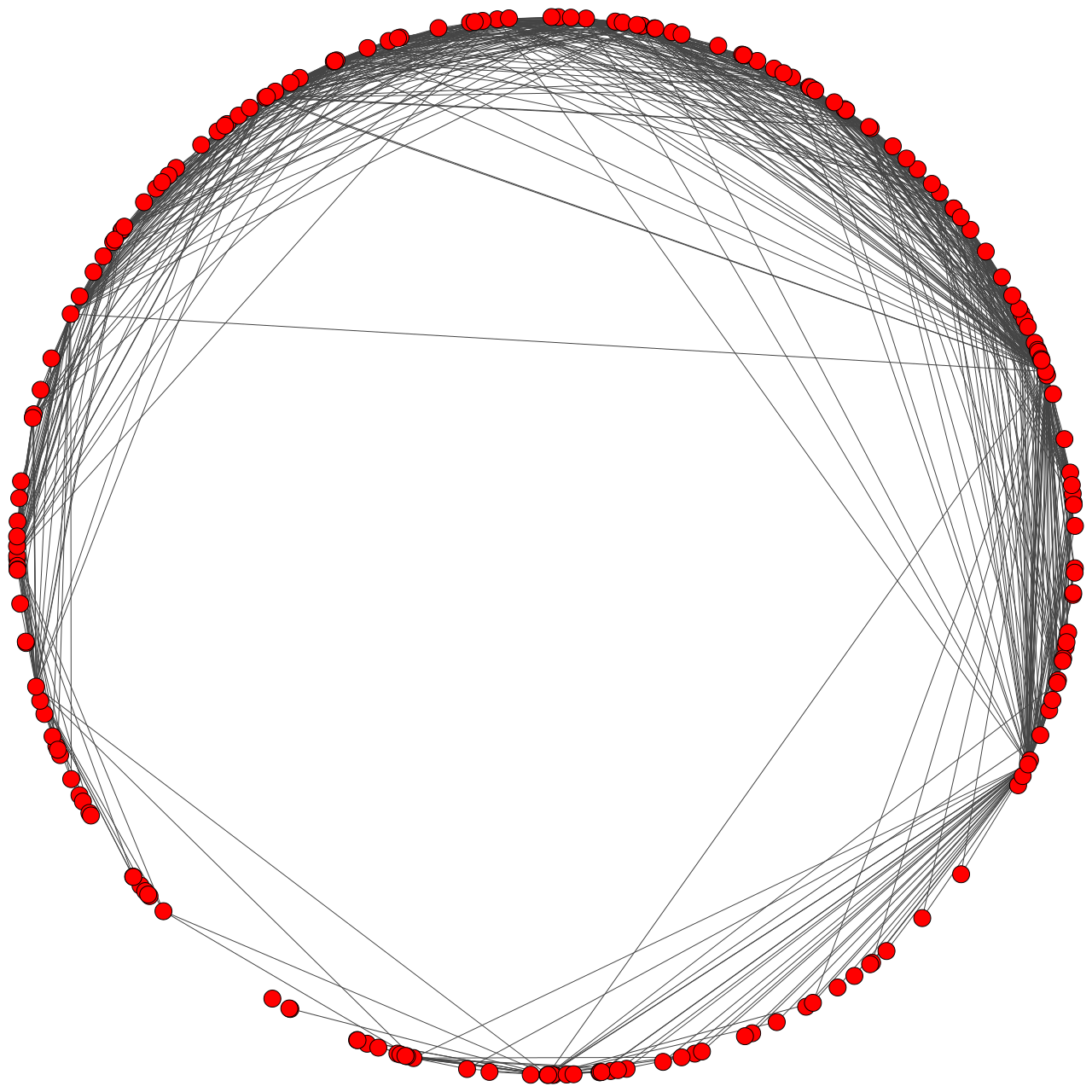}
   \end{subfigure}
   \begin{subfigure}[b]{0.3\textwidth}
       \includegraphics[width=\textwidth]{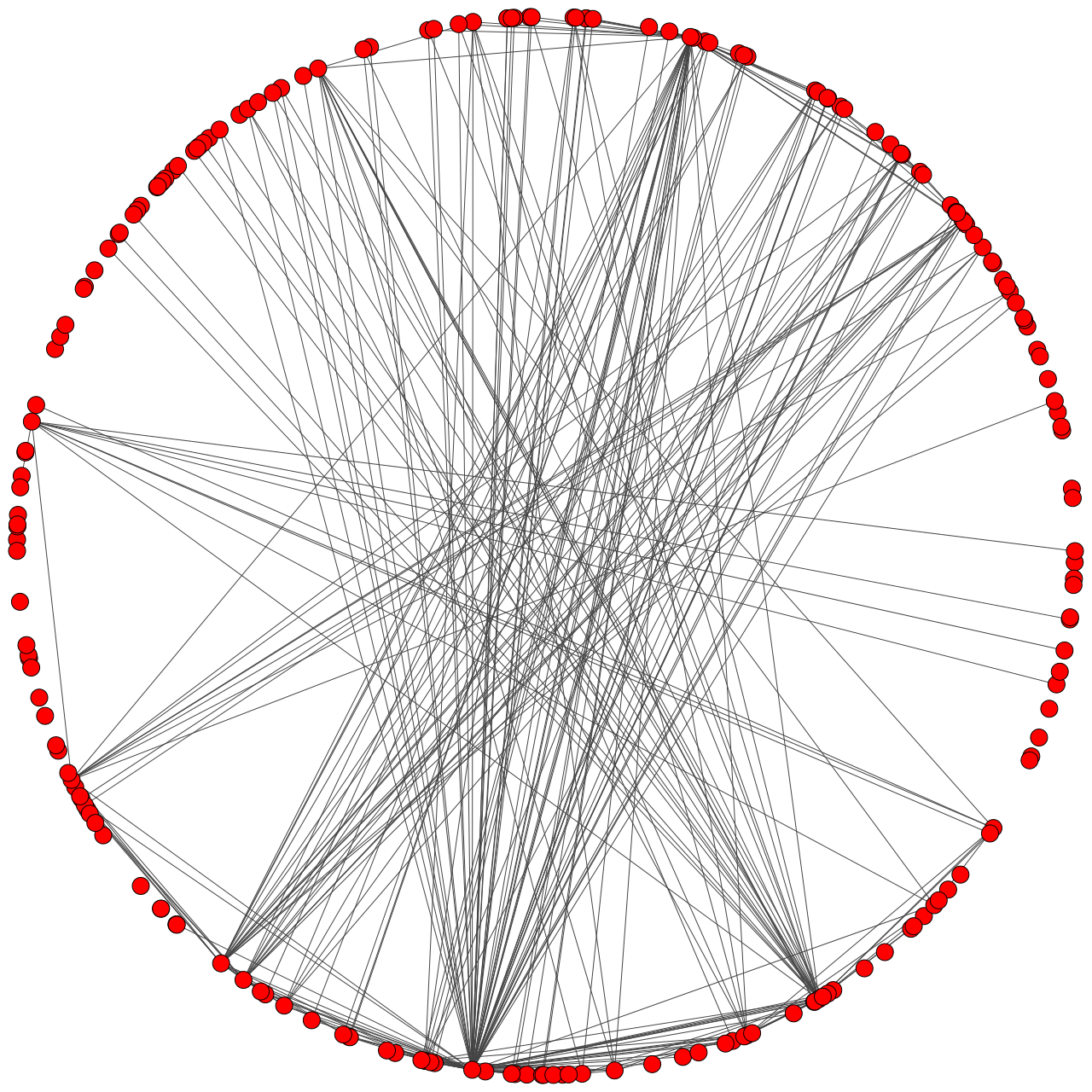}
   \end{subfigure}
   \caption{Example in $\mathbb{R}^2$ of influence $h(x) = x$, $h(x) = e^x$, $h(x) = x^2$}\label{fig:r2dif}
\end{figure}

\begin{theorem}\label{th:fotdpdd}
	Let $P_{out}(k)$ be the probability of an random node to have out-degree $k$, $P_{in}(k)$ – in-degree $k$. If ${n^{\max\{\alpha, \beta\}/a}/\theta(n) = o(1)}$ and $h:[-1, 1] \to \mathbb{R}$ - continuous, strictly increasing function, positive at least in one point from $(-1, 1)$, then there exist constants $C_0$ and $N_0$ such that $\forall k(n):\forall n > N_0 \,\, k(n) < C_0n$ we have
$$
	P_{out}(k) = k^{-1 - \alpha/\beta}(1+o(1)),
	P_{in}(k)  = k^{-1 - \beta/\alpha}(1+o(1)).
$$

\end{theorem}

\begin{proof}[Short scheme of proof] Here is the scheme of proof for the out-degree distribution. The case of the in-degree is similar.

Restrictions on the function $h$ allow us to modify the proof of the directed case. The main difference is a value of the probability $P_e(w)$ of a node $\vec{v_i} = w_i \vec{x_i}$ with the weight $w_i = w$ to have an edge to another random node.

\begin{equation}\label{eq:h_p_e(w)}
	P_e(w) = \int_{w_0}^{\infty} \frac{aw_0^a}{(w')^{a+1}}\int_{\substack{x' \in S^2 \\ w^\alpha (w')^\beta h((x,x'))\geq\theta}} \frac{1}{4\pi} \mathrm{d}x' \mathrm{d}w'.
\end{equation}

We will denote by $I$ the inner integral:
\begin{equation}\int_{\substack{x' \in S^2 \\ w^\alpha (w')^\beta h((x,x'))\geq\theta}} \frac{1}{4\pi} \mathrm{d}x' \mathrm{d}w'.\end{equation}

We can rewrite inequality~(\ref{ineq:h_condition}) as $ h((x, x')) \geq \frac{\theta}{w^\alpha (w')^\beta}$ and notice that $\frac{\theta}{w^\alpha (w')^\beta} \in (0, +\infty)$. Let us consider $h([-1, 1]) = [r, q]$, on this interval function $h$ is invertable. We examine the mutual position of $[r, q]$ and $(0, +\infty)$. The definition of $h$ implies that $[r, q] \cap (0, +\infty) \neq \emptyset$. This gives us the next two cases.

\textbf{A)} The first case is $[r, q] \subset (0, +\infty)$. If $\frac{\theta}{w^\alpha (w')^\beta} \in [r, q]$, then we may invert $h$ and the inner integral $I$ is equal to $2\pi\left(1 - h^{-1}\left(\frac{\theta}{w^\alpha (w')^\beta}\right)\right)$. If $\frac{\theta}{w^\alpha (w')^\beta} > q$, than the inequality~(\ref{ineq:h_condition}) is not satisfied and $I=0$. If $0 < \frac{\theta}{w^\alpha (w')^\beta} < r$, than the inequality~(\ref{ineq:h_condition}) is satisfied for any pair of $x$ and $x'$, $I = 4\pi$, the surface area of $S^2$.

To deal with $P_e(w)$, we need to compare $w_0$ with boundaries for each range of $\frac{\theta}{w^\alpha (w')^\beta}$.

1) If $w_0 < \frac{\theta^{1/\beta}}{w^{\alpha/\beta} q^{1/\beta}}$, then 
\begin{gather*}
	P_e(w) = \int_{w_0}^{\frac{\theta^{1/\beta}}{w^{\alpha/\beta} q^{1/\beta}}} 0 \mathrm{d}w' + \int_{\frac{\theta^{1/\beta}}{w^{\alpha/\beta} q^{1/\beta}}}^{\frac{\theta^{1/\beta}}{w^{\alpha/\beta} r^{1/\beta}}} \frac{aw_0^a}{(w')^{a+1}} \frac{1}{2} \big[1 -  h^{-1}\big(\frac{\theta}{w^\alpha (w')^\beta}\big) \big] \mathrm{d}w' + \\ +
    \int_{\frac{\theta^{1/\beta}}{w^{\alpha/\beta} r^{1/\beta}}}^{\infty} 4\pi \frac{aw_0^a}{(w')^{a+1}} \mathrm{d}w'.
\end{gather*}

2) If $\frac{\theta^{1/\beta}}{w^{\alpha/\beta} q^{1/\beta}} \leq w_0 < \frac{\theta^{1/\beta}}{w^{\alpha/\beta} r^{1/\beta}}$, then

\begin{gather*}
	P_e(w) = \int_{w_0}^{\frac{\theta^{1/\beta}}{w^{\alpha/\beta} r^{1/\beta}}} \frac{aw_0^a}{(w')^{a+1}} \frac{1}{2} \big[1 -  h^{-1}\big(\frac{\theta}{w^\alpha (w')^\beta}\big) \big] \mathrm{d}w' + \\ + 
    \int_{\frac{\theta^{1/\beta}}{w^{\alpha/\beta} r^{1/\beta}}}^{\infty} 4\pi \frac{aw_0^a}{(w')^{a+1}} \mathrm{d}w'.
\end{gather*}

3) Last case is $w_0 \geq  \frac{\theta^{1/\beta}}{w^{\alpha/\beta} r^{1/\beta}}$. But $\theta(n)$ grows with $n$ and for big enough $n$ this inequality will not be satisfied.




\textbf{B)} The second case is $[r, q]  \not\subset (0, +\infty)$, which implies $r \leq 0$. If $\frac{\theta}{w^\alpha (w')^\beta} \in (0, q]$, then $I=2\pi\left(1 - h^{-1}\left(\frac{\theta}{w^\alpha (w')^\beta}\right)\right)$. If $\frac{\theta}{w^\alpha (w')^\beta} > q$, then $I=0$.  This gives

\begin{gather*}
	P_e(w) = \int_{\max(w_0, \frac{\theta^{1/\beta}}{w^{\alpha/\beta} q^{1/\beta}})}^{\infty} \frac{aw_0^a}{(w')^{a+1}} \frac{1}{2} \big[1 -  h^{-1}\big(\frac{\theta}{w^\alpha (w')^\beta}\big) \big] \mathrm{d}w'
\end{gather*}

It remains only to show that $P_{out}(k) = k^{-2}(1+o(1))$.
But now it is easy to see that the influnce of every kind of the principal parts of the integral for $P_e(w)$ has been already examined in previous theorems for degree distributions. For example,
\begin{gather*}
\int_{\frac{\theta^{1/\beta}}{w^{\alpha/\beta} q^{1/\beta}}}^{\frac{\theta^{1/\beta}}{w^{\alpha/\beta} r^{1/\beta}}} \frac{aw_0^a}{(w')^{a+1}} \frac{1}{2} \big[1 -  h^{-1}\big(\frac{\theta}{w^\alpha (w')^\beta}\big) \big] \mathrm{d}w' =\\= \frac{w_0^a w^{2a\alpha/\beta}}{\beta \theta^{a/\beta}}\int_{r}^{q} (1 - h^{-1}(t)) t^{a/\beta - 1} \mathrm{d}t,
\end{gather*}
what is proportional to the one we got in Theorem~\ref{th:ordd}. Therefore we are not giving here additional details.  




\end{proof}

For example, described class of functions contains functions like $e^x$ and ${x^{2m+1} + c}$,  ${m \in \mathbb{N}}$, for a proper constant $c$.

Of course, not only this small class of functions $h(x)$ has no influence on the degree distribution. For example, it is easy to show that $h(x) = x^{2m}, m \in \mathbb{N}$ also has this property. In this way, a proof will be different only in the computation of $P_e(w)$. 

\section{Conclusion}

\indent \indent In our work we suggest a new model for scale-free networks generation, which is based on the matrix factorization and has a geographical interpretation.
We formalize it for fixed size and growing networks. We proof and validate empirically that degree distribution of resulting networks obeys power-law with an exponent of 2.

We also consider several extensions of the model. First, we research the case of the directed network and obtain power-law degree distribution with a tunable exponent. Then, we apply different functions to the dot product of latent features vectors, which give us modifications with interesting properties.

Further research could focus on the deep study of latent features vectors distribution.  It seems that not only a uniform distribution over the surface of the sphere should be considered because, for example, cities are not uniformly distributed over the surface of Earth. Besides, we want to try other distributions of weights.

\section{Appendix}
\subsection{Proof of Lemma \ref{lm:probWithWeight}}
For a node $x$ with the weight $w$, the probability to be connected to a random node is represented by
\begin{equation}
\label{eq:p-w-e-int}
P_e(w) = \int_{w_0}^{\infty}f(w')\int_{\substack{x' \in S^2 \\ ww'(x,x')\geq\theta}} \frac{1}{4\pi} \mathrm{d}x' \mathrm{d}w' .
\end{equation}

We can rewrite inequality $ww'(x,x')\geq\theta$ as  ${(x,x')\geq\frac{\theta}{ww'}}$. If $\frac{\theta}{ww'} \in [0, 1]$, this inequality defines the spherical cap  of the area $2\pi(1 - \frac{\theta}{ww'})$.
Therefore, we have 
\begin{equation}
\label{eq:p-w-e-2}
P_e(w) = \int_{\max\{w_0, \theta / w\}}^{\infty}f(w')2\pi\left(1 - \frac{\theta}{ww'}\right)\frac{1}{4\pi}  \mathrm{d}w' .
\end{equation}
If we substitute $f(w')$ from (\ref{eq:pareto}), we obtain
\begin{equation}
\label{eq:p-w-e-3}
P_e(w) = \int_{\max\{w_0, \theta / w\}}^{\infty}\frac{a}{w_0}\left(\frac{w_0}{w'}\right)^{a+1} \frac{1}{2} \left(1 - \frac{\theta}{ww'}\right) \mathrm{d}w' .
\end{equation}
If $w \leq \theta / w_0$, then
\begin{gather*}
 P_e(w) = \int_{\theta / w}^{\infty}\frac{a}{2w_0}\left(\frac{w_0}{w'}\right)^{a+1} \left(1 - \frac{\theta}{ww'}\right)  \mathrm{d}w' 
 = \\ 
 = \int_{\theta / w}^{\infty}\frac{a}{2w_0}\left(\frac{w_0}{w'}\right)^{a+1}   \mathrm{d}w' 
 - \int_{\theta / w}^{\infty}\frac{a}{2w_0}\left(\frac{w_0}{w'}\right)^{a+1} \frac{\theta}{ww'}  \mathrm{d}w'  
 =\\ 
 = \frac{aw_0^a}{2} \frac{1}{a \left(\theta / w\right)^{a}} - \frac{aw_0^a\theta}{2w} \frac{1}{(a+1) (\theta / w)^{a+1}} = \frac{1}{2}\frac{w_{0}^a}{\theta^a (a + 1)} w^a . 
\end{gather*} 
 If $w > \theta / w_0$, then
 \begin{gather*}
 P_e(w) = \int_{w_0}^{\infty}\frac{a}{w_0}\left(\frac{w_0}{w'}\right)^{a+1} 2\pi\left(1 - \frac{\theta}{ww'}\right)\frac{1}{4\pi}  \mathrm{d}w' =\\ 
 = \frac{aw_0^a}{2} \int_{w_0}^{\infty}\frac{1}{w'^{a+1}} \mathrm{d}w' - \frac{aw_0^a\theta}{2w} \int_{ w_0}^{\infty}\frac{1}{w'^{a+2}}  \mathrm{d}w' = \\ =
 \frac{aw_0^a}{2} \frac{1}{a w_0^{a}} - \frac{aw_0^a\theta}{2w} \frac{1}{(a+1) w_0^{a+1}} 
 = \frac{1}{2}\left(1 - \frac{a\theta}{w(a+1)w_{0}}\right).
\end{gather*}

\subsection{Proof of Lemma \ref{lm:expectedNumberEdges}}
The edge probability is represented by
\begin{equation}
P_e = \int_{w_0}^{\infty}\int_{S^2}\int_{w_0}^{\infty} \int_{\substack{x' \in S^2 \\ ww'(x,x')\geq\theta}} f(w) f(w') \frac{1}{16\pi^2} \mathrm{d}x' \mathrm{d}w' \mathrm{d}x \mathrm{d}w.
\end{equation}
Using (\ref{eq:p-w-e-int}), we obtain
\begin{equation}
P_e = \int_{w_0}^{\infty} \int_{S(0,1)} \frac{1}{4\pi} f(w) P_e(w)  \mathrm{d}x \mathrm{d}w =  \int_{w_0}^{\infty} f(w) P_e(w) \mathrm{d}w .
\end{equation}
If $\theta < w_{0}^2$, then for all $w \in [w_0, \infty)$   $P_e(w)$ equals to $\frac{1}{2}(1 - \frac{a\theta}{w(a+1)w_{0}})$. Using it, we get
\begin{gather*}
P_e = \int_{w_0}^{\infty} \frac{1}{2}(1 - \frac{a\theta}{w(a+1)w_{0}}) a \frac{w_0^a}{w^{a+1}}  \mathrm{d}w = \\
  = \frac{1}{2} -  \int_{w_0}^{\infty} \frac{1}{2}(\frac{a\theta}{w(a+1)w_{0}}) a \frac{w_0^a}{w^{a+1}}  \mathrm{d}w 
  = \frac{1}{2} -  \frac{1}{2} a^2 \theta \frac{w_0^{a-1}}{a+1} \int_{w_0}^{\infty} \frac{1}{w^{a+2}}  \mathrm{d}w = \\
= \frac{1}{2} -  \frac{1}{2} a^2 \theta \frac{w_0^{a-1}}{a+1} \frac{1}{a+1}\frac{1}{w_0^{a+1}} =
\frac{1}{2} -  \frac{1}{2} \frac{a^2}{(a+1)^2}\frac{\theta}{w_0^{2}}  . 
\end{gather*}

If $\theta \geq w_{0}^2$, then

\begin{gather*}
P_e = \int_{w_0}^{\theta/w_0} \frac{1}{2}\frac{w_{0}^a}{\theta^a (a + 1)} w^a a \frac{w_0^a}{w^{a+1}} \mathrm{d}w   +  \int_{\theta/w_0}^{\infty} \frac{1}{2}(1 - \frac{a\theta}{w(a+1)w_{0}}) a \frac{w_0^a}{w^{a+1}}  \mathrm{d}w  =\\
=\frac{1}{2}\frac{w_{0}^a}{\theta^a (a + 1)}aw_0^a \int_{w_0}^{\theta/w_0}\frac{1}{w} \mathrm{d}w  + \frac{1}{2} a w_0^{a} \int_{\theta/w_0}^{\infty} \frac{1}{w^{a+1}} - \\
- \frac{a^2w_0^{a-1} \theta}{2(a+1)} \int_{\theta/w_0}^{\infty} \frac{1}{w^{a+2}} \mathrm{d}w = \\ 
= \frac{1}{2}\frac{w_{0}^{2a} a}{\theta^a (a + 1)} (\ln \theta - 2\ln w_0)  + \frac{w_0^{2a}}{2\theta^a} - \frac{a^2}{2(a+1)^2}\frac{w_0^{2a}}{\theta^{a}} .
\end{gather*}

\subsection{Proof of Lemma \ref{th:variation}}

Let us enumerate pairs of nodes. Each pair of nodes $i$ has an edge indicator~$I_{e_i}$.

By definition, we have
\begin{gather*}
\mathrm{Var}(M) = \mathbb{E}(M^2) - \mathbb{E}(M)^2  = \mathbb{E}(I_{e_1} + \ldots + I_{e_{n(n-1)/2}})^2 - (\mathbb{E}I_{e_1} + \ldots + \\ + \mathbb{E}I_{e_{n(n-1)/2}})^2 
= \sum_{i} \mathbb{E}I_{e_i}^2 + 2\sum_{i\neq j} \mathbb{E}I_{e_i}I_{e_j} - \sum_{i} (\mathbb{E}I_{e_i})^2 - 2\sum_{i\neq j}\mathbb{E}I_{e_i}\mathbb{E}I_{e_{j}} .
\end{gather*}

$I_{e_1}$, $\ldots$, $I_{e_{n(n-1)/2}}$ is the sequence of identically distributed random variables, so their expected value is the same and equals to $P_e$.

Since $\mathbb{E}I_{e_i}^2 = \mathbb{E}I_{e_i} = P_e$, it follows that
\begin{gather*}
\mathbb{E}I_{e_i}I_{e_j} - \frac{n(n-1)}{2}(P_e)^2 - 2\sum_{i\neq j}\mathbb{E}I_{e_i}\mathbb{E}I_{e_j}  = \\
=  \frac{n(n-1)}{2} P_e(1 -P_e) + 2\sum_{i\neq j} \mathbb{E}I_{e_i}I_{e_j} - 2\sum_{i\neq j}\mathbb{E}I_{e_i}\mathbb{E}I_{e_j}.
\end{gather*}

If edges $e_{i}$ and $e_j$ do not have mutual nodes, then $I_{e_i}$ and $I_{e_j}$ are independent variables. Therefore, $\mathrm{E}(I_{e_i} I_{e_j}) = \mathrm{E}(I_{e_i}) \mathrm{E}(I_{e_j}) = P_e^2$. We get

\begin{gather*}
 \mathrm{Var}(M) = \frac{n(n-1)}{2} P_e(1 - P_e) + \\
 + \sum_{v = 1}^{n} \sum_{\substack{w = 1 \\ w \neq v}}^{n} \sum_{\substack{z = w + 1\\ z\neq v}}^{n}  
 (\mathbb{E}I_{e(v, w)}I_{e(v, z)} - \mathbb{E}I_{e(v, w)}\mathbb{E}I_{e(v, z)}) = \\
 = \frac{n(n-1)}{2} P_e(1 - P_e) + \sum_{v = 1}^{n} \sum_{\substack{w = 1 \\ w \neq v}}^{n} \sum_{\substack{z = w + 1\\ z\neq v}}^{n}  
 (\mathbb{E}I_{e(v, w)}I_{e(v, z)} - P_e^2)
\end{gather*}

$\mathbb{E}I_{e(v, w)}I_{e(v, z)}$ is exactly equal to $P_<$.

\subsection{Proof of Lemma \ref{lm:P_triangle}}

It can be easily seen that
\begin{equation*}
P_{<} = \int_{w_0}^{\infty} P_e(w)^2 f(w) \mathrm{d}w.
\end{equation*}

If $\theta < w_{0}^2$ we have

\begin{gather*}
P_< = \int_{w_0}^{\infty} \frac{1}{4}\Big(1 - \frac{a\theta}{w(a+1)w_{0}}\Big)^2 a \frac{w_0^a}{w^{a+1}}  \mathrm{d}w = \frac{1}{4} a w_0^a \int_{w_0}^{\infty} \frac{1}{w^{a+1}}  \mathrm{d}w - \\
 - \frac{1}{2} \frac{a^2 \theta w_0^{a-1}}{a+1} \int_{w_0}^{\infty} \frac{1}{w^{a+2}}  \mathrm{d}w +
\frac{1}{4}\frac{a^3 \theta^2 w_0^{a-2} }{(a+1)^2} \int_{w_0}^{\infty} \frac{1}{w^{a+3}}  \mathrm{d}w = \\
= \frac{1}{4}  -
\frac{1}{2} \frac{a^2 \theta }{(a+1)^2} \frac{1}{w_0^2} +
\frac{1}{4}\frac{a^3 \theta^2 }{(a+1)^2(a+2)} \frac{1}{w_0^4}.
\end{gather*}

If $\theta \geq w_{0}^2$, then

\begin{gather*}
P_< =  \int_{w_0}^{\theta/w_0} \frac{1}{4}\frac{w_{0}^{2a}}{\theta^{2a} (a + 1)^2} w^{2a} a \frac{w_0^a}{w^{a+1}} \mathrm{d}w   + \\
+ \int_{\theta/w_0}^{\infty} \frac{1}{4}\Big(1 - \frac{a\theta}{w(a+1)w_{0}}\Big)^2 a \frac{w_0^a}{w^{a+1}}  \mathrm{d}w.
\end{gather*}

Computing the first integral, we get 

\begin{gather*}
\int_{w_0}^{\theta/w_0} \frac{1}{4}\frac{w_{0}^{2a}}{\theta^{2a} (a + 1)^2} w^{2a} a \frac{w_0^a}{w^{a+1}} \mathrm{d}w = \frac{1}{4}\frac{w_{0}^{2a}}{\theta^{2a}(a + 1)^2}a w_0^{a} \int_{w_0}^{\theta/w_0} w^{a-1} \mathrm{d}w = \\
= \frac{1}{4}\frac{w_0^{2a}}{\theta^{2a}(a+1)^2} [\theta^{a} - w_0^{2a}].
\end{gather*}

And for the second one we have
\begin{gather*}
\int_{\theta/w_0}^{\infty} \frac{1}{4}\Big(1 - \frac{a\theta}{w(a+1)w_{0}}\Big)^2 a \frac{w_0^a}{w^{a+1}}  \mathrm{d}w = \\ 
=
\int_{\theta/w_0}^{\infty} \frac{1}{4} a \frac{w_0^a}{w^{a+1}}  \mathrm{d}w - 
\int_{\theta/w_0}^{\infty} \frac{1}{2}\frac{a\theta}{w(a+1)w_{0}} a \frac{w_0^a}{w^{a+1}}  \mathrm{d}w + \\
+ \int_{\theta/w_0}^{\infty} \frac{1}{4}\frac{a^2\theta^2}{w^2(a+1)^2w_{0}^2} a \frac{w_0^a}{w^{a+1}}  \mathrm{d}w = 
\frac{1}{4} a w_0^a \int_{\theta/w_0}^{\infty} \frac{1}{w^{a+1}}  \mathrm{d}w - \\
- \frac{1}{2} \frac{a^2 \theta w_0^{a-1}}{a+1} \int_{\theta/w_0}^{\infty} \frac{1}{w^{a+2}}  \mathrm{d}w + 
\frac{1}{4}\frac{a^3 \theta^2 w_0^{a-2} }{(a+1)^2} \int_{\theta/w_0}^{\infty} \frac{1}{w^{a+3}}  \mathrm{d}w = \\
 = 
\frac{1}{4}  w_0^a \frac{w_0^{a}}{\theta^{a}} - 
\frac{1}{2} \frac{a^2 \theta w_0^{a-1}}{(a+1)^2} \frac{w_0^{a+1}}{\theta^{a+1}} +
\frac{1}{4}\frac{a^3 \theta^2 w_0^{a-2} }{(a+1)^2(a+2)} \frac{w_0^{a+2}}{\theta^{a+2}} = \\  
= \frac{1}{4}\frac{w_0^{2a}}{\theta^{a}} -
\frac{1}{2} \frac{a^2 }{(a+1)^2} \frac{w_0^{2a}}{\theta^{a}} +
\frac{1}{4}\frac{a^3   }{(a+1)^2(a+2)} \frac{w_0^{2a}}{\theta^{a}}.
\end{gather*}

This gives us $P_<$ in the case of $\theta \geq w_{0}^2$:

\begin{gather*}
P(<) = \frac{1}{4}\frac{w_0^{2a}}{\theta^{2a}(a+1)^2} [\theta^{a} - w_0^{2a}] + \frac{1}{4}\frac{w_0^{2a}}{\theta^{a}}\Big[ 1 - 
2 \frac{a^2 }{(a+1)^2} +
\frac{a^3   }{(a+1)^2(a+2)} \Big].
\end{gather*}


\begin{backmatter}

\section*{Competing interests}
The authors declare that they have no competing interests.

\section*{Author's contributions}
This work is the result of a close joint effort in which all authors contributed almost equally to defining and shaping the problem definition, proofs, algorithms, and manuscript. The research would not have been conducted without the participation of any of the authors. All authors participated in writing article, read and approved the final manuscript.


\bibliographystyle{bmc-mathphys} 
\bibliography{references}      







\end{backmatter}
\end{document}